\documentclass[%
 aip,
 jmp,onecolumn,eqsecnum,
 amsmath,amssymb,
 reprint,%
author-numerical,%
]{revtex4-2}

\usepackage[utf8]{inputenc}
\usepackage[T1]{fontenc}
\usepackage{etoolbox,hyperref}

\makeatletter
\def\@email#1#2{%
 \endgroup
 \patchcmd{\titleblock@produce}
  {\frontmatter@RRAPformat}
  {\frontmatter@RRAPformat{\produce@RRAP{#1\href{mailto:#2}{#2}}}\frontmatter@RRAPformat}
  {}{}
}%
\makeatother

\def\LEFTRIGHT#1#2#3{\left#1 #3 \right#2}

\newcommand{\HH}{\mathbb{H}}
\newcommand{\R}{\mathbb{R}}
\newcommand{\N}{\mathbb{N}}
\newcommand{\C}{\mathbb{C}}
\newcommand{\Z}{\mathbb{Z}}

\renewcommand{\le}{\leqslant}
\renewcommand{\ge}{\geqslant}

\renewcommand{\d}{\mathrm{d}}  

\let\emptyset\varnothing
\let\comp\circ

\let\Re\undefined
\let\Im\undefined
\DeclareMathOperator{\Re}{Re}
\DeclareMathOperator{\Im}{Im}

\DeclareMathOperator{\e}{e}
\DeclareMathOperator{\tr}{tr}

\usepackage{amsthm}

\makeatletter

\newtheoremstyle{ttheorem}%
       {1.8ex\@plus1ex}                
       {2.1ex\@plus1ex\@minus.5ex}      
       {\itshape}           
       {0pt}                   
       {\bfseries}          
       {.}                  
       {.5em}               
       {}                

\newtheoremstyle{ddefinition}%
       {1.8ex\@plus1ex}                
       {2.1ex\@plus1ex\@minus.5ex}      
       {}           
       {0pt}                   
       {\bfseries}           
       {.}                  
       {.5em}               
       {}                

\newtheoremstyle{rremark}%
       {1.8ex\@plus1ex}                
       {2.1ex\@plus1ex\@minus.5ex}      
       {\normalfont}        
       {0pt}                   
       {\bfseries}           
       {.}                  
       {.5em}               
       {}                   

\theoremstyle{ttheorem}
\newtheorem{theorem}{Theorem}
\newtheorem{lemma}[theorem]{Lemma}

\newtheorem{corollary}[theorem]{Corollary}

\theoremstyle{ddefinition}
\newtheorem{definition}[theorem]{Definition}
\newtheorem{ass}[theorem]{Assumption}

\theoremstyle{remark}
\newtheorem{remark}[theorem]{Remark}
\newtheorem{myremarks}[theorem]{Remarks}

\newenvironment{remarks}{\begin{myremarks}\begin{nummer}}%
    {\end{nummer}\end{myremarks}}

\newcounter{numcount}
\newcommand{\labelnummer}{(\roman{numcount})}%

\providecommand{\showkeyslabelformat}[1]{\relax}        
\let\mysaveformat\showkeyslabelformat                   %
\def\myformat#1{\raisebox{-1.5ex}{\mysaveformat{#1}}}   %

\newenvironment{nummer}%
  {\let\curlabelspeicher\@currentlabel%
    \begin{list}{\textup{\labelnummer}}%
      {\usecounter{numcount}\leftmargin0pt%
        \topsep0.5ex\partopsep2ex\parsep0pt\itemsep0ex\@plus1\p@%
        \labelwidth2.5em\itemindent3.5em\labelsep1em%
      }%
    \let\saveitem\item%
    \def\item{\saveitem%
      \def\@currentlabel{\curlabelspeicher\kern.1em\labelnummer}}%
    \let\savelabel\label%
    \def\label##1{{\ifnum\thenumcount=1\let\showkeyslabelformat\myformat\fi\savelabel{##1}}%
										{\def\@currentlabel{\labelnummer}%
									 	\let\showkeyslabelformat\@gobble
									 	\savelabel{##1item}%
										}%
	   							}%
  }{\end{list}}%

\makeatother

\begin{document}

\title{Stability of a Szeg\H{o}-type asymptotics}
\thanks{This paper is dedicated to Abel Klein in honour of his numerous fundamental contributions to mathematical physics. As a mentor and friend of the first-named author, Abel has been a constant source of inspiration to him.}

\author{Peter M\"uller}%
\affiliation{Mathematisches Institut, Ludwig-Maximilians-Universit\"at M\"unchen,
  Theresienstra\ss{e} 39, 80333 M\"unchen, Germany}
\email[Author to whom correspondence should be addressed: ]{mueller@lmu.de}

\author{Ruth Schulte}%
\affiliation{Mathematisches Institut, Ludwig-Maximilians-Universit\"at M\"unchen,
  Theresienstra\ss{e} 39, 80333 M\"unchen, Germany}
\date{\today}

\begin{abstract}
	We consider a multi-dimensional continuum Schr\"odinger operator $H$ which is given by a perturbation 
	of the negative Laplacian by a compactly supported bounded potential. 
	We show that, for a fairly large class of test functions, the second-order Szeg\H{o}-type asymptotics for the spatially truncated Fermi projection of $H$ is independent of the potential and, thus, identical to the known asymptotics of the Laplacian. 
\end{abstract}

\maketitle

%
\section{Introduction and Result}
%

A classical result of Szeg\H{o} describes the asymptotic growth of the determinant of a truncated 
Toeplitz matrix as the truncation parameter tends to infinity.\cite{Szego15, Szego52}
Jump discontinuities in the spectral function (or symbol) of the Toeplitz matrix are of crucial 
importance for the growth of the subleading term in such an asymptotic expansion.\cite{FiHa69, Basor86}

Recent years have witnessed considerable interest in Szeg\H{o}-type asymptotics for spectral 
projections of Schr\"odinger operators.\cite{Wolf:2006ek, LeschkeSobolevSpitzer14,%
	PhysRevLett.113.150404, ElgartPasturShcherbina2016, Dietlein18,
	PfirschSobolev18, MuPaSc19, MuSc20, CharlesEst20, LeSoSp21, Pfeiffer21, PfSp22}
The first mathematical proof of such an asymptotics was established by Leschke, Sobolev and Spitzer
\cite{LeschkeSobolevSpitzer14} and relies on extensive work 
of Sobolev,\cite{MR3076415, Sobolev:2014ew} making rigorous a long-standing conjecture of 
Widom.\cite{Widom:1982uz}
The authors of Ref.~\onlinecite{LeschkeSobolevSpitzer14} consider the simplest and most 
prominent Schr\"odinger operator,  
the self-adjoint (negative) Laplacian 
$H_{0}:= -\Delta$. It acts on a dense domain in the Hilbert space $L^{2}(\R^{d})$ of complex-valued 
square-integrable functions over $d$-dimensional Euclidean space $\R^{d}$. 
More precisely, given a (Fermi) energy $E>0$, they  consider the spectral 
projection $1_{<E}(H_{0})$ of $H_{0}$ associated with the interval ${}]-\infty, E[\,$. 
Besides physical motivation, this spectral projection provides a prototypical example for a symbol with 
a single discontinuity. The spectral projection $1_{<E}(H_{0})$ gives rise to a Wiener--Hopf operator by
truncation with the multiplication operator $1_{\Lambda_{L}}$ from the left and from the right. 
The truncation corresponds to multiplication with the indicator function of the spatial subset
$\Lambda_{L} := L \cdot \Lambda \subset\R^{d}$, $L>0$, which is the scaled version of some
``nice'' bounded subset $\Lambda \subset \R^{d}$. We state the details for $\Lambda$ in 
Assumption~\ref{ass:l}(i) below.
The final ingredient is a ``test function'' $h : [0,1] \rightarrow\C$, which is piecewise continuous 
and vanishes at zero, $h(0) =0$. Furthermore it is required that $h$ grows
at most algebraically near the endpoints of the interval, that is, there exists a ``H\"older exponent'' $\alpha >0$ such that  
\begin{equation}
 	h(\lambda) = O(\lambda^{\alpha}) \quad\text{and}\quad 
	h(1) - h(1-\lambda) = O(\lambda^{\alpha}) \quad\text{as}\quad  \lambda \searrow 0.
\end{equation}
Here, we employed the Bachmann--Landau notation for asymptotic equalities: 
We will use the big-$O$ and little-$o$ symbols throughout the paper. 
Under these hypotheses, the second-order asymptotic formula for the trace
\begin{equation}
	\label{eq:widom}
 	\tr\big\{h\big(1_{\Lambda_L}1_{<E}(H_0)1_{\Lambda_L}\big)\big\} 
	=N_{0}(E) h(1)|\Lambda|L^{d} + \Sigma_{0}(E) I(h) \, |\partial\Lambda| L^{d-1} \ln L 
	+ {o}\big(L^{d-1} \ln L\big) 
\end{equation}
as $L\to\infty$ is proved in Ref.~\onlinecite{LeschkeSobolevSpitzer14}. 
Here, $|\Lambda|$ denotes the (Lebesgue) volume of $\Lambda$ and $|\partial\Lambda|$ denotes the surface area of the 
boundary $\partial\Lambda$ of $\Lambda$.
The leading-order coefficient in \eqref{eq:widom} is determined by the integrated density of states
\begin{equation}
	\label{eq:n0}
 	N_{0}(E) := \frac{1}{\Gamma[(d+1)/2]} \, 
	\bigg(\frac{ E}{4\pi}\bigg)^{d/2},
\end{equation}
of $H_{0}$. Here, $\Gamma$ denotes Euler's gamma function.
The coefficient of the subleading term factorises into a product of 
\begin{equation}
	\label{eq:i-of-h}
	I(h):= \frac{1}{4\pi^{2}} \int_0^1 \d\lambda \,\frac{h(\lambda) - \lambda h(1)}{\lambda(1-\lambda)}
\end{equation}
and
\begin{equation}
	\label{eq:sigma0}
 	\Sigma_{0}(E)  := \frac{2}{\Gamma[(d+1)/2]} \, 
	\bigg(\frac{ E}{4\pi}\bigg)^{(d-1)/2}.
\end{equation}
We follow the usual terminology and refer to \eqref{eq:widom} as a (second-order) Szeg\H{o}-type asymptotics.
The occurrence of the logarithmic factor $\ln L$ multiplying the surface area $L^{d-1}$ in 
the subleading term of \eqref{eq:widom} is attributed to the discontinuity of the symbol $1_{<E}$ 
together with the \emph{dynamical} delocalisation of the Schr\"odinger 
time evolution generated by the Laplacian.\cite{MuPaSc19}
In the context of entanglement entropies for non-interacting Fermi gases, see \eqref{eq:EE} below, the
occurrence of this additional $\ln L$-factor is also coined an 
\emph{enhanced area law}.\cite{LeschkeSobolevSpitzer14}

A natural question concerns the fate of the asymptotics \eqref{eq:widom} when $H_{0} = - \Delta$ 
is replaced by a  general self-adjoint Schr\"odinger operators $H:=-\Delta+V$ with 
a (suitable) electric potential $V$.
Unfortunately, there exists no general approach which allows to derive a two-term asymptotics 
like \eqref{eq:widom} for 
$H$. Mathematical proofs are restricted to special examples or classes of examples. 
The exact determination of the coefficient in the subleading term beyond bounds poses a particularly challenging task.

First, we describe two situations in which -- in contrast to \eqref{eq:widom} -- the subleading term does not 
exhibit a logarithmic enhancement. Such a behaviour is generally referred to as an \emph{area law} and 
caused an enormous attraction in the physics literature over several decades until now, 
see e.g.\ Refs.~\onlinecite{RevModPhys.82.277, Laflorencie:2016kg}, and references therein. 
An area law is typically expected if $H$ has a mobility gap in its spectrum and 
if the Fermi energy falls inside the mobility gap. 
It was proved for discrete random Schr\"odinger operators, a 
Fermi energy lying in the region of complete localisation and for test functions obeying some smoothness assumption,
including the ones for entanglement entropies in \eqref{eq:h} below.%
\cite{PhysRevLett.113.150404, ElgartPasturShcherbina2016, MR3744386}
An area law also shows up in the Szeg\H{o} asymptotics for rather general test functions when $H$ equals the Landau 
Hamiltonian in two dimensions with a perpendicular constant magnetic field and if the Fermi energy $E$ 
coincides with one of the Landau levels,\cite{LeSoSp21} see also Ref.~\onlinecite{CharlesEst20}. 
Very recently, Ref.~\onlinecite{Pfeiffer21} established the stability of this area law under suitable magnetic and 
electric perturbations and for the test functions \eqref{eq:h} corresponding to entanglement entropies. 

Now, we return to situations where the subleading term exhibits a logarithmic enhancement 
as in \eqref{eq:widom}. One-dimensional Schr\"odinger operators $H= -\Delta + V$ with an 
arbitrarily often differentiable periodic potential $V$ were studied in 
Ref.~\onlinecite{PfirschSobolev18}. 
If the Fermi energy lies in the interior of a Bloch band, 
Pfirsch and Sobolev\cite{PfirschSobolev18} establish an enhanced area law for the same class of test functions as considered 
in \eqref{eq:widom}. Surprisingly, the coefficient of the subleading term of the 
Szeg\H{o} asymptotics remains the same as for $H_{0}=-\Delta$ in \eqref{eq:widom}. 
However, as the integrated density of states of the periodic Schr\"odinger operator $H$ differs from that of $H_{0}$ 
in general, this affects the leading-order term of the asymptotics. 
In contrast to the area law in the two-dimensional Landau model,\cite{LeSoSp21} a logarithmic enhancement occurs when this model is considered in three space dimensions due to the free motion in the direction parallel to the magnetic field.\cite{PfSp22} 
Another situation was studied by the present authors in Ref.~\onlinecite{MuSc20}. There, $H_{0}$ is perturbed 
by a compactly supported and bounded potential. The enhanced area law for $H_{0}$ is then proven to persist for 
the test function $h=h_{1}$ from \eqref{eq:h}, which corresponds to the von Neumann entropy. 
Yet the bounds in Ref.~\onlinecite{MuSc20}  are not good enough as to allow for a conclusion concerning the coefficient. 

The purpose of this paper is to show that the second-order Szeg\H{o} asymptotics \eqref{eq:widom} remains valid 
with the same coefficients if $H_{0}$ is replaced by $H=-\Delta+V$ with a compactly supported and bounded 
potential $V\in L^\infty(\R^d)$. This improves the result in Ref.~\onlinecite{MuSc20} in two ways: 
(i) the statement is strengthened  as to cover also universality of the coefficients and 
(ii) it is extended from the von Neumann entropy to a fairly large class of test functions.
The general approach we follow here is different from that in Ref.~\onlinecite{MuSc20}. 
It combines the traditional way \cite{LaWi80} for proving \eqref{eq:widom}, see also e.g.\ 
Refs.~\onlinecite{LeschkeSobolevSpitzer14, PfirschSobolev18}, with improved estimates from 
Ref.~\onlinecite{MuSc20}.

Our assumptions on the spatial domain coincide with those in Ref.~\onlinecite{MuSc20}.
\begin{ass}
	\label{ass:l}
	We consider a bounded Borel set $\Lambda\subset\R^{d}$ such that 
	\begin{enumerate}
 	\item[(i)]  it is a Lipschitz domain with, if $d\ge 2$, a piecewise $C^{1}$-boundary, 
	\item[(ii)] the origin $0\in \R^{d}$ is an interior point of $\Lambda$.
	\end{enumerate}
\end{ass}

\begin{remark}
	Assumption~\ref{ass:l}(i) is taken from Ref.~\onlinecite{LeschkeSobolevSpitzer14} 
	and guarantees the validity of
	\eqref{eq:widom},  see also Cond.\ 3.1 in Ref.~\onlinecite{LeschkeSobolevSpitzer17} for the notion of 
	a Lipschitz domain. Assumption~\ref{ass:l}(ii) 
	can always be achieved by a translation of the potential $V$ in Theorem~\ref{thm:main}.
\end{remark}

We specify the set of test functions to which our main result applies.

\begin{definition}
	\label{HH-def}
	For $d\in\N\setminus\{1\}$ we set
	\begin{multline*}
 	 	\HH_{d} := 
		\Big\{ h : [0,1] \rightarrow\C \text{~ piecewise continuous, $h(0) =0$ 
				and $\exists\, \alpha > d^{-1}$ such that} \\[-1ex] 
				\text{$h(\lambda) = O(\lambda^{\alpha})$ and
				$h(1) - h(1-\lambda) = O(\lambda^{\alpha})$ as $\lambda \searrow 0$} \Big\}.
	\end{multline*}
	In $d=1$ space dimension we require test functions to have an additional mirror symmetry 
	at $\lambda =1/2$ and to vanish	faster at zero (and one) than linear times a logarithm
	\begin{multline*}
 	 	\HH_{1} := 
		\Big\{ h : [0,1] \rightarrow\C \text{~ piecewise continuous, $h(0) =0$, $h = h(1- \boldsymbol\cdot)$ 
				and} \\[-1ex]
				\text{$h(\lambda) = o(\lambda\ln\lambda)$ as $\lambda \searrow 0$} \Big\}.
	\end{multline*}
\end{definition}

The main result of this paper is about the first two terms in the asymptotics \eqref{eq:main-form}.
These terms coincide with the ones in \eqref{eq:widom} 
and, thus, do not depend on the potential $V$.

\begin{theorem}
	\label{lem:themoregeneralthing}
	Let $d\in\N$ and let $\Lambda\subset\R^{d}$ be as in Assumption~{\upshape\ref{ass:l}}. 
	Let $V\in L^\infty(\R^d)$ be a compactly supported potential and let $h\in \HH_{d}$ be a test function. 
	Then, for every Fermi energy $E>0$ we obtain
	\begin{equation}
		\label{eq:main-form}
	 	\tr\big\{h\big(1_{\Lambda_L}1_{<E}(H)1_{\Lambda_L}\big)\big\} 
		= N_{0}(E) h(1)|\Lambda|L^{d} + \Sigma_{0}(E) I(h) |\partial\Lambda| L^{d-1} \ln L 
			+ {o}\big(L^{d-1} \ln L\big) 
	\end{equation}	 
	as $L\to\infty$.
\end{theorem}

\begin{remarks}
\item 
	The restriction to symmetric test functions in $\HH_{1}$ is technical. 
	It relates to the incommodious fact that the trace of a sequence of operators may converge whereas 
	convergence in trace norm need not hold. We refer to Remark~\ref{rem:1Dfail} for more details.
\item
	The class of test functions for which \eqref{eq:widom} holds requires less regularity for $h$ near the endpoints of the interval $[0,1]$ as compared to $\HH_{d}$. We expect that Theorem~\ref{lem:themoregeneralthing} extends to this more general class, i.e.\ to test functions $h$ with arbitrarily small ``H\"older exponents'' $\alpha$. 
	 Possibly, such an extension requires additional smoothness of the potential $V$. 
\item 
	It is only for the sake of simplicity that we confined ourselves to compactly supported and bounded potentials. 
	The theorem should remain valid for potentials with sufficient integrability properties. 
\end{remarks}

We conclude this section with an application of Theorem~\ref{lem:themoregeneralthing} to entanglement entropies for
the ground state of a system of non-interacting fermions with single-particle Hamiltonian $H$.
We introduce the one-parameter family of test functions 
$h_{\alpha}:\,[0,1]\rightarrow[0,1]$, given by 
\begin{equation}
	\label{eq:h} 
	\lambda\mapsto h_\alpha(\lambda):= \begin{cases} 
		(1-\alpha)^{-1}\log_{2}[\lambda^\alpha+(1-\lambda)^\alpha], 
			& \text{if } \;\alpha\in{}]0,\infty[\,\setminus\{1\}, \\
		-\lambda\log_2 \lambda - (1-\lambda)\log_2(1-\lambda), 
			&  \text{if } \;\alpha=1, 
	\end{cases}
\end{equation}
and refer to them as R\'enyi entropy functions. The particular case $h_{1}$ is also called von Neumann entropy function, for which
we use the convention $0 \log_{2}0 :=0$ for the binary logarithm. Following Ref.~\onlinecite{Klich06}, the quantity 
\begin{equation}\label{eq:EE}
	S_{\alpha}(H,E,\Omega) :=\tr\big\{h_{\alpha}\big(1_{\Omega}1_{<E}(H)1_{\Omega}\big)\big\},
\end{equation}
$\alpha>0$, defines the R\'enyi-entanglement entropy with respect to a spatial bipartition for the ground state 
of a quasi-free Fermi gas characterised by the 
single-particle Hamiltonian $H$ and Fermi energy $E$. Here, $\Omega \subset\R^{d}$ is any bounded Borel set.
It is obivous from the definition of the space of test functions that 
\begin{equation}
	h_{\alpha} \in \HH_{d} \quad\text{if and only if} \quad  d > \alpha^{-1} .
\end{equation}
Therefore, Theorem~\ref{lem:themoregeneralthing} has the following immediate

\begin{corollary}\label{thm:main}
	Let $d\in\N$ and let $\Lambda\subset\R^{d}$ be as in Assumption~{\upshape\ref{ass:l}}. Let
	$V\in L^\infty(\R^d)$ be a compactly supported potential. We fix a Fermi energy $E>0$ and a R\'enyi index 
	$\alpha > d^{-1}$.
	Then, the R\'enyi entanglement entropy exhibits the enhanced area law
	\begin{equation}
		\label{eq:entropy-asy}
		\lim_{L\rightarrow\infty}\frac{S_{\alpha}(H,E,\Lambda_L)}{L^{d-1}\ln L}
		= \Sigma_{0}(E) I(h_{\alpha}) |\partial\Lambda|,
	\end{equation}
	the limit being independent of $V$.
\end{corollary}

\begin{remark}
 	Again, it would be desirable to remove the restriction $\alpha > d^{-1}$ in the corollary.
\end{remark}

%
\section{Proof of Theorem~\ref{lem:themoregeneralthing}}
%

Theorem~\ref{lem:themoregeneralthing} follows from \eqref{eq:widom}, which was proven in
Ref.~\onlinecite{LeschkeSobolevSpitzer14}, and 

\begin{theorem}
	\label{thm:diff}
	Let $d\in\N$ and let $\Lambda\subset\R^{d}$ be as in Assumption~{\upshape\ref{ass:l}}. 
	Let $V\in L^\infty(\R^d)$ be a compactly supported potential and let $h\in \HH_{d}$ be a test function. 
	Then, for every Fermi energy $E>0$ we have
	\begin{equation}
		\label{eq:widom-diff}
	 	\tr\Big\{h\big(1_{\Lambda_L}1_{<E}(H)1_{\Lambda_L}\big) - h\big(1_{\Lambda_L}1_{<E}(H_{0}) 1_{\Lambda_L}\big)\Big\} 
		= {o}(L^{d-1} \ln L) \qquad\text{as $L\to\infty$.}
	\end{equation}  
\end{theorem}

\noindent
The following notion will be useful in the proof of Theorem~\ref{thm:diff}. 

\begin{definition}
 	For every $d\in\N$ we set 
	\begin{equation}
		\HH_{d,0} := \{h \in \HH_{d}: h(1)=0\}.
 	\end{equation}
	We note that $\HH_{1,0} = \HH_{1}$ due to the additional symmetry constraint in $d=1$.
\end{definition}


\begin{proof}[Proof of Theorem~{\upshape\ref{thm:diff}}]
	We fix $E>0$ and let $L>1$. 
 	For a given $h\in \HH_{d}$, we have $\tilde{h}:= h -h(1) \,\mathrm{id} \in \HH_{d,0}$. 
	Throughout this paper, we use the notation $H_{(0)}$ as a placeholder for either $H$ or $H_{0}$ 
	and similarly for other quantities like 
	\begin{equation}
 		P_{L (, 0)} := 1_{\Lambda_{L}} 1_{<E}(H_{(0)})1_{\Lambda_{L}}.
	\end{equation}
	By definition of $\tilde h$, we observe
	\begin{equation}
 		h(P_{L}) - h(P_{L,0}) = \tilde h(P_{L}) - \tilde h(P_{L,0}) + h(1) (P_{L}- P_{L,0}).
	\end{equation}
	We recall that if $d=1$ then $h(1)=0$ by symmetry.
	Thus, Lemma~\ref{lem:P-diff-1norm} below implies that the proof of the theorem is reduced to 
	showing the claim for test functions from $\HH_{d,0}$ only. 
	This will be accomplished in four steps for test functions of increasing generality.
	
	\noindent\emph{Step \textup{(i)}. $h\in \HH_{d,0}$ is a polynomial.} \\
	We define even and odd polynomials $s_{n}$ and $a_{n}$, $n\in\N$, on $[0,1]$ by 
	\begin{equation}
		\label{eq:polys}
 		s_{n}(\lambda) := [\lambda(1-\lambda)]^{n} \quad \text{and} \quad
		a_{n}(\lambda) := \lambda s_{n}(\lambda), \quad \lambda\in [0,1],
	\end{equation}
	and note that the family $\{s_{n}, a_{n}: n\in \N\}$ constitutes a basis of the space of polynomials in $\HH_{d,0}$ for $d \ge 2$ because the linear spans
		\begin{equation}
			\label{eq:span-equal}
 			\mathrm{span}\{s_{n}, a_{n}: n\in \N\} = \mathrm{span}\{s_{1}\mkern1mu\mathrm{id}^{\,k}: k\in \N_{0}\}
		\end{equation}
		coincide. For $d=1$, only the symmetric polynomials are to be considered for the basis due to the symmetry constraint. Therefore, the claim of Step (i) amounts to proving 
	\begin{align}
 		\tr\big\{ s_{n}(P_{L}) - s_{n}(P_{L,0})\big\} &= {o}(L^{d-1} \ln L) \label{eq:sym-pol} \\
		\intertext{and, if $d\ge 2$, also } 
		\tr\big\{ a_{n}(P_{L}) - a_{n}(P_{L,0})\big\} &= {o}(L^{d-1} \ln L) \label{eq:asym-pol}
	\end{align}	
	as $L\to\infty$ for every $n\in \N$. 

	We turn to \eqref{eq:sym-pol} first and observe that $s_{1}(P_{L(, 0)}) = |Q_{L(, 0)}|^{2}$ with
	\begin{equation}
		\label{eq:Qdef}
		Q_{L(, 0)} := 1_{\Lambda_L^{c}}1_{<E}(H_{(0)})	1_{\Lambda_L},
	\end{equation}
	where $|A|^2:=A^\ast A$ for any bounded operator $A$, 
	and the superscript $^{c}$ indicates the complement of a set. 
	The telescope-sum identity
	\begin{equation}
 		a^{n} = b^{n} + \sum_{j=1}^{n} a^{n-j}(a-b) b^{j-1}
	\end{equation}
	for $a,b\in\R$ and $n\in\N$ does not use the commutativity of $a$ and $b$ and, thus, implies
	\begin{equation}
		\label{eq:sn-start}
 		s_{n}(P_{L}) = \big(|Q_{L}|^{2}\big)^{n} 
		= s_{n}(P_{L,0}) + \sum_{j=1}^{n} |Q_{L}|^{2(n-j)} \big( |Q_{L}|^{2} - |Q_{L,0}|^{2}\big) |Q_{L,0}|^{2(j-1)}.
	\end{equation}
	Hence, we obtain 
	\begin{equation}
		\label{eq:sn-middle}
 		\big| \tr\big\{ s_{n}(P_{L}) - s_{n}(P_{L,0}) \big\} \big| \le n \big\| |Q_{L}|^{2} - |Q_{L,0}|^{2} \big\|_{1},
	\end{equation}
	where $\|\boldsymbol\cdot\|_{p}$, $p>0$, denotes the von Neumann--Schatten (quasi-) norm.
	We infer from the identity 
	\begin{equation}
 		|Q_{L}|^{2} - |Q_{L,0}|^{2} = (Q_{L}^{*} - Q_{L,0}^{*})(Q_{L} - Q_{L,0}) + (Q_{L}^{*} - Q_{L,0}^{*})Q_{L,0}
			+ Q_{L,0}^{*}(Q_{L} - Q_{L,0})
	\end{equation}
	and \eqref{eq:sn-middle} that 
	\begin{equation}
		\label{eq:sn-ende}
 		\big| \tr\big\{ s_{n}(P_{L}) - s_{n}(P_{L,0}) \big\} \big| \le n \big( \|Q_{L} - Q_{L,0}\|_{2}^{2}
			+ 2 \|Q_{L} - Q_{L,0}\|_{2} \|Q_{L,0}\|_{2}\big) =: n  \phi(L).
	\end{equation}
	The boundedness of Hilbert--Schmidt norms
	\begin{equation}
		\label{eq:HSbounded}
 		\sup_{L>1} \|Q_{L} - Q_{L,0} \|_{2} < \infty,
	\end{equation}
	see Lemma~2.3 in Ref.~\onlinecite{MuSc20},
	and the asymptotics \eqref{eq:widom} applied to the test function $s_{1}$,
	\begin{equation}
 		\|Q_{L,0}\|_{2} = [ \tr s_{1}(P_{L,0}) ]^{1/2} = O\big( L^{(d-1)/2} (\ln L)^{1/2}\big) \quad
		\text{as } L\to\infty,
	\end{equation}
	yield 
	\begin{equation}
		\label{eq:phi-L}
 		\phi(L) = O\big( L^{(d-1)/2} (\ln L)^{1/2}\big) \quad	\text{as } L\to\infty.
	\end{equation}
	Therefore, \eqref{eq:sym-pol} follows from \eqref{eq:sn-ende} and \eqref{eq:phi-L}.

	Now, we turn to the proof of \eqref{eq:asym-pol} and equate, using \eqref{eq:sn-start},
	\begin{equation}
 		a_{n}(P_{L}) = a_{n}(P_{L,0}) + (P_{L} - P_{L,0}) s_{n}(P_{L,0}) 
			+ P_{L} \sum_{j=1}^{n} |Q_{L}|^{2(n-j)} \big( |Q_{L}|^{2} - |Q_{L,0}|^{2}\big) |Q_{L,0}|^{2(j-1)}.
	\end{equation}
	This leads to the estimate
	\begin{align}
 		\big| \tr\big\{ a_{n}(P_{L}) - a_{n}(P_{L,0}) \big\} \big| 
		&\le \big|\tr\big\{(P_{L} - P_{L,0}) s_{n}(P_{L,0})\big\}\big|	+ n \phi(L)  \notag\\
		&\le \|P_{L} - P_{L,0}\|_{2} \,\|s_{n}(P_{L,0})\|_{2} + n \phi(L).  \label{eq:norm-est-bad}
	\end{align}
	The first factor of the first term in the second line of \eqref{eq:norm-est-bad} is of the order $O(\sqrt{\ln L})$ as 
	$L\to\infty$ according to Lemma~\ref{lem:P-diff-2norm} below. The second factor of the first term
	behaves like 
	\begin{equation}
 		\|s_{n}(P_{L,0})\|_{2} 
		= \big(\tr\{ s_{2n}(P_{L,0}) \}\big)^{1/2} 
		= O\big( L^{(d-1)/2} (\ln L)^{1/2}\big) \qquad	\text{as } L\to\infty,
	\end{equation}
	according to \eqref{eq:widom} applied to the test function $s_{2n}$.
	Together with \eqref{eq:phi-L} we conclude from \eqref{eq:norm-est-bad} that 
	\begin{equation}
 		\big| \tr\big\{ a_{n}(P_{L}) - a_{n}(P_{L,0}) \big\} \big| = O(L^{(d-1)/2} \ln L)
		\qquad	\text{as } L\to\infty,
	\end{equation}
	which proves \eqref{eq:asym-pol} for $d\ge 2$.

	Steps~(ii) -- (iv) establish the ``closure'' of the asymptotics.
	
	\noindent\emph{Step \textup{(ii)}. $h\in \HH_{d,0}$ is of the form $h=s_{1}f$ with $f\in C([0,1])$.} \\
	Let $\varepsilon>0$. The Stone--Weierstra\ss{} theorem guarantees the existence of a polynomial 
	$\zeta$ over $[0,1]$ such that
	\begin{equation}
		\label{eq:StoneW}
		\sup_{\lambda\in [0,1]}|f(\lambda) - \zeta(\lambda) | 	\le \varepsilon .
	\end{equation}
	We estimate with the triangle inequality 
	\begin{align}
 		\big|\tr \big\{ (s_{1}f)(P_{L}) - (s_{1}f)(P_{L,0})\big\} \big|
		& \le \big\|\big(s_{1}(f-\zeta)\big)(P_{L})\big\|_{1} 
			+ \big\|\big(s_{1}(\zeta -f)\big)(P_{L,0})\big\|_{1} 	\notag\\		
		& \quad 	+ \big|\tr \big\{ (s_{1}\zeta)(P_{L}) - (s_{1}\zeta)(P_{L,0})\big\} \big| \notag \\
		& \le \varepsilon \big( \|s_{1}(P_{L})\|_{1} + \|s_{1}(P_{L,0})\|_{1} \big) + o(L^{d-1}\ln L)
	\end{align}
	as $L\to\infty$, where the last estimate uses H\"older's inequality, \eqref{eq:StoneW} and the result of 
	Step~(i) for the polynomial $s_{1}\zeta \in \HH_{d,0}$. 
	We observe $ \|s_{1}(P_{L(,0)})\|_{1} = \tr s_{1}(P_{L(,0)})$ and conclude with \eqref{eq:widom} plus 
	another application of Step~(i) to the difference $ \tr s_{1}(P_{L}) -  \tr s_{1}(P_{L,0})$ that 
	\begin{equation}
 		\lim_{L\to\infty} \frac{\big|\tr \big\{ (s_{1}f)(P_{L}) - (s_{1}f)(P_{L,0})\big\} \big|}{L^{d-1}\ln L}
		\le 2 \varepsilon \Sigma_{0}(E)I(s_{1})|\partial\Lambda|.
	\end{equation}
	As $\varepsilon>0$ is arbitrary, the claim follows.

	\noindent\emph{Step \textup{(iii)}. $h\in \HH_{d,0}$ is continuous.} \\
	Boundedness of $h$ and the growth conditions at $0$ and $1$ provide the existence of a constant $C>0$ 
	and of a function $g: [0,1/4] \rightarrow [0, \infty[\,$ with $g(0)=0$ and $\lim_{x\searrow 0}g(x)=0$
	such that 
	\begin{equation}
 		|h| \le \LEFTRIGHT\{.{\begin{array}{cl} C s_{1}^{\alpha}, & \text{if \; $d \ge 2$ \; with \; $\alpha > 1/d$}\\ 
						- s_{1} \log_{2}(s_{1}) \, g(s_{1}), & \text{if \; $d =1$}. \end{array}}
	\end{equation}
	For $\varepsilon \in \,]0,1/2[\,$ arbitrary but fixed, 
	we consider a non-negative switch function $\tilde\psi_{\varepsilon} \in C^{\infty}([0,1])$ with 
	$0 \le\tilde\psi_{\varepsilon} \le 1$ 
	and with restrictions
	\begin{equation}
 		\tilde\psi_{\varepsilon}\big|_{[0,\varepsilon/2]} = 1 \qquad\text{and}\qquad
		\tilde\psi_{\varepsilon}\big|_{[\varepsilon, 1]} = 0.
	\end{equation}
	We set $\psi_{\varepsilon} := \tilde\psi_{\varepsilon}\comp s_{1}$. 
	Since $(1-\psi_{\varepsilon})h$ is of the form of the previous Step~(ii), we conclude
	\begin{align}
 		\frac{\big|\tr \big\{ h(P_{L}) - h(P_{L,0})\big\} \big|}{L^{d-1}\ln L}
		& \le \frac{\big|\tr \big\{ (\psi_{\varepsilon}h)(P_{L}) - (\psi_{\varepsilon}h)(P_{L,0})\big\} 
					\big|}{L^{d-1}\ln L} + o (1) \notag\\
	 	&	\le \frac{\|(\psi_{\varepsilon}h)(P_{L}) \|_{1}}{L^{d-1}\ln L} + \Sigma_{0}(E) 
			I(|\psi_{\varepsilon}h|) 	|\partial\Lambda| +o(1)
		\label{eq:psi-eps-diff}
	\end{align}
	as $L\to\infty$,
	where we used the Szeg\H{o} asymptotics \eqref{eq:widom} for the unperturbed operator $P_{L,0}$.

	First, we consider the case $d\ge 2$. Let $\delta \in \, ]0, \alpha - 1/d[\,$ so that 
	$\alpha':=\alpha - \delta > 1/d$.	We observe that 
	\begin{equation}
		\label{eq:psi-h-bound}
		|\psi_{\varepsilon}h| \le C s_{1}^{\alpha} \psi_{\varepsilon} \le C \varepsilon^{\delta} s_{1}^{\alpha'},
	\end{equation}	 
	whence
	\begin{equation}
		\label{eq:psi-h-bound2}		
 		\|(\psi_{\varepsilon}h)(P_{L})\|_{1}	\le C \varepsilon^{\delta} 
			\|Q_{L}\|_{2\alpha'}^{2\alpha'} 
		\le C_{\alpha'} \varepsilon^{\delta} \Big(\|Q_{L,0}\|_{2\alpha'}^{2\alpha'} 
			+ \|Q_{L}- Q_{L,0}\|_{2\alpha'}^{2\alpha'} \Big). 
	\end{equation}
	Here, $C_{\alpha'} \ge C$ is a constant needed to cover the case $2\alpha' >1$. The difference term satisfies
	$\|Q_{L}- Q_{L,0}\|_{2\alpha'}^{2\alpha'} = o(L^{d-1})$ as $L\to\infty$ by Corollary~\ref{cor:interpolation} if 
	$\alpha' \in \, ]1/d,1]$. For the remaining case $\alpha' >1$, we refer to the general von Neumann--Schatten-norm
	estimate $\|\boldsymbol\cdot\|_{2\alpha'} \le \|\boldsymbol\cdot\|_{2}$ and \eqref{eq:HSbounded}. 
	Thus, we infer from \eqref{eq:psi-eps-diff} that
	\begin{align}
 		\lim_{L\to\infty}\frac{\big|\tr \big\{ h(P_{L}) - h(P_{L,0})\big\} \big|}{L^{d-1}\ln L}
		& \le \Sigma_{0}(E) I(|\psi_{\varepsilon}h|) |\partial\Lambda|  
				+ C_{\alpha'} \varepsilon^{\delta} 	\!\lim_{L\to\infty}
				\frac{\tr \{ s_{1}^{\alpha'}(P_{L,0})\}}{L^{d-1}\ln L}	\notag\\
		& = \Sigma_{0}(E) |\partial\Lambda|  \, \big[  I(|\psi_{\varepsilon}h|) 
			+  C_{\alpha'} 	\varepsilon^{\delta} I(s_{1}^{\alpha'})	\big],
		\label{eq:psi-eps-diff2}
	\end{align}
	where we (ab-)used the notation \eqref{eq:polys} for the symmetric polynomials to 
	include also positive real exponents and 
	appealed to the unperturbed asymptotics \eqref{eq:widom} in the last step.
	The bound \eqref{eq:psi-h-bound} implies 
	$I(|\psi_{\varepsilon}h|) \le C \varepsilon^{\delta} I(s_{1}^{\alpha'}) < \infty$
	so that the claim follows from \eqref{eq:psi-eps-diff2} and the fact that $\varepsilon>0$ 
	can be chosen arbitrarily small.
	
	It remains to treat the case $d=1$. In this case, \eqref{eq:psi-h-bound} and \eqref{eq:psi-h-bound2} 
	are to be replaced by 
	\begin{equation}
	\label{eq:psi-h-bound-d1} 		
	|\psi_{\varepsilon}h| \le  - s_{1} \log_{2}(s_{1}) g(s_{1}) \psi_{\varepsilon} 
		\le - s_{1} \log_{2}(s_{1})  \psi_{\varepsilon} G(\varepsilon),		
	\end{equation}
	where $G(\varepsilon) := \sup_{x\in [0,\varepsilon]} g(x)$, and 
	\begin{equation}
 		\|(\psi_{\varepsilon}h)(P_{L})\|_{1}	
		\le G(\varepsilon) \, \big\| |Q_{L}|^{2} \log_{2}(|Q_{L}|^{2}) \big\|_{1} 
		= G(\varepsilon) \tr\big\{f(|Q_{L}|)\big\},
	\end{equation}
	where $f:[0,\infty[\, \rightarrow [0,1]$, $x \mapsto  1_{[0,1]}(x) (-x^{2}) \log_{2}(x^{2})$. 
 	The ``Proof of the upper bound in Theorem~1.3'' in Ref.~\onlinecite{MuSc20} demonstrates
	\begin{equation}
		\label{eq:MS-ub}
 		\tr\big\{f(|Q_{L}|)\big\} = O(\ln L) \qquad \text{as } L\to\infty,
	\end{equation}
	see Eq.\ (2.49) in Ref.~\onlinecite{MuSc20}.
	Thus, we conclude with \eqref{eq:psi-eps-diff} and \eqref{eq:psi-h-bound-d1} -- \eqref{eq:MS-ub} 
	that there exists a constant $c>0$ such that
	\begin{align}
 		\lim_{L\to\infty}\frac{\big|\tr \big\{ h(P_{L}) - h(P_{L,0})\big\} \big|}{\ln L}
		& \le c G(\varepsilon) + \Sigma_{0}(E) I(|\psi_{\varepsilon}h|) |\partial\Lambda|  \notag\\
		& \le G(\varepsilon) \big[ c + \Sigma_{0}(E) I(s_{1}|\log_{2}s_{1}|) |\partial\Lambda| \big]  .
		\label{eq:psi-eps-diff1}
	\end{align}
	The claim follows	because of $I(s_{1}|\log_{2}s_{1}|) < \infty$ and 
	$\lim_{\varepsilon\downarrow 0} G(\varepsilon) =0$ , which is consequence of the right-continuity of $g$.

	\noindent\emph{Step \textup{(iv)}. $h\in \HH_{d,0}$.} \\
	We follow the argument in Ref.~\onlinecite{PfirschSobolev18} and assume -- without restriction -- that $h$ is real-valued.
	Otherwise we decompose $h$ into its real part $\Re h$ and imaginary part $\Im h$.
	
	Since $h$ is piecewise continuous and continuous at $0$ and at $1$ there exists $\delta>0$ such that 
	$h$ is continuous in $[0,2\delta] \cup [1-2\delta,1]$. Let $\varepsilon>0$. We choose
	continuous functions $h_{1}, h_{2} \in \HH_{d,0} \cap C([0,1])$ such that 
	\begin{enumerate}
 	\item[(1)] $h$, $h_{1}$ and $h_{2}$ coincide on $[0,\delta] \cup [1-\delta,1]$,
	\item[(2)] $h_{1} \le h \le h_{2}$,
	\item[(3)] $\int_{0}^{1} \d\lambda\, | h_{2}(\lambda) - h_{1}(\lambda)| < \varepsilon$.
	\end{enumerate}
	The monotonicity (2) and an application of Step~(iii) to $h_{1}$, respectively $h_{2}$, imply
	\begin{equation}
 	 \frac{\tr \big\{ h_{1}(P_{L,0}) - h(P_{L,0})\big\}}{L^{d-1}\ln L}	+ o(1) 
	 \le \frac{\tr \big\{ h(P_{L}) - h(P_{L,0})\big\}}{L^{d-1}\ln L}  
	 \le \frac{\tr \big\{ h_{2}(P_{L,0}) - h(P_{L,0})\big\}}{L^{d-1}\ln L}	+ o(1)
	\end{equation}
	as $L\to\infty$.
	 The unperturbed asymptotics \eqref{eq:widom} for $h_{1} - h$, respectively $h_{2}- h$, thus gives
	\begin{equation}
		\label{eq:iv-end}
 	 \Sigma_{0}(E)I(h_{1}-h)|\partial\Lambda|	 \le \lim_{L\to\infty}\frac{\tr \big\{ h(P_{L}) - h(P_{L,0})\big\}}{L^{d-1}\ln L} \le
	 \Sigma_{0}(E)I(h_{2}-h)|\partial\Lambda|.	
	\end{equation}
	By (1) and (3) we estimate $I(|h_{2}-h_{1}|) \le \eta\varepsilon$ with $\eta:=1/[4\pi^{2}\delta(1-\delta)]$. Combining this with the monotonicity (2), we infer
	\begin{equation}
 		-\eta\varepsilon \le I(h_{1}-h) \qquad \text{and}\qquad I(h_{2}-h) \le \eta\varepsilon.
	\end{equation}
	Since $\varepsilon>0$ is arbitrary, the claim follows from \eqref{eq:iv-end}.
\end{proof}

\begin{lemma}
 	\label{lem:P-diff-2norm}
	Assume the hypotheses of Theorem~{\upshape\ref{thm:diff}}. Then we have 
	\begin{equation}
 		\|P_{L} - P_{L,0} \|_{2} = O(\sqrt{\ln L}) \qquad \text{as } L \to \infty.
	\end{equation}
\end{lemma}

\begin{proof}
	Let $\Gamma_{x} := x + [0,1]^{d} \subset\R^{d}$ denote the closed unit cube translated by $x\in\R^{d}$.
	By Eq.~(2.25) in Ref.~\onlinecite{MuSc20}, there exist $c_{2} \equiv c_{2}(E,V)> 0$ and $\ell \equiv \ell (E,V) >0$ such that 
	\begin{equation}
		\label{eq:MS2.25}
 		\big\|1_{\Gamma_{n}}\big( 1_{<E}(H) - 1_{<E}(H_{0})\big) 1_{\Gamma_{m}}\big\|_{2}
		\le \frac{c_{2}}{(|n||m|)^{(d-1)/2} (|n| + |m|)} 
	\end{equation}
	for all $n,m \in \Z^{d} \setminus [-\ell,\ell]^{d}$. In order to deal also with centres in the excluded cube 
	$[-\ell,\ell]^{d}$ we use the estimate
	\begin{equation}
 		\|A\|_{2}^{2} \le 3 \|1_\Omega A\|_{2}^{2} + \| 1_{\Omega^{c}} A 1_{\Omega^{c}} \|_{2}^{2}
	\end{equation}
	for any self-adjoint Hilbert--Schmidt operator $A$ and any measurable subset $\Omega \subset \R^{d}$, 
	together with the norm bound  
	\begin{equation}
		\label{eq:pNorm-bound}
	 	\sup_{x\in\R^{d}} \| 1_{\Gamma_{x}} 1_{<E}(H_{(0)})\|_{p} \le \gamma,
	\end{equation}
	which holds for any $p>0$, see Lemma~B.1 in Ref.~\onlinecite{DiGeMu19}. Here, $\gamma \equiv \gamma(E,p,V) > 0$ is a constant. 
	As to the applicability of Lemma~B.1 in Ref.~\onlinecite{DiGeMu19}, we remark that the result of that lemma is 
	formulated with two spatial indicator functions: one to the left and one to the right of the spectral projection.
	Yet, only one of them is needed in the proof. The other one drops out
	in Eq.\ (155) of Ref.~\onlinecite{DiGeMu19}.

	Combining \eqref{eq:MS2.25} -- \eqref{eq:pNorm-bound}, we infer the existence of constants 
	$C \equiv C(\ell, E, V)>0$, $C' \equiv C'(\ell, E, V)>0$ and $a \equiv a(\Lambda) >0$, 
	which are independent of $L$, such that 	
	\begin{align}
		\|P_{L} - P_{L,0} \|_{2}^{2} 
		& \le C + \sum_{\substack{n,m\in\Z^{d} \setminus [-\ell,\ell]^{d}: \\[.5ex] \Gamma_{n} 
			\cap \Lambda_{L} \neq \emptyset,		\, \Gamma_{m} \cap \Lambda_{L} \neq \emptyset}}  
			\frac{c_{2}^{2}}{(|n||m|)^{d-1} (|n| + |m|)^{2}} \notag\\
		& \le C' \int_{\ell}^{aL} \d x \int_{\ell}^{aL} \d y \; \frac{1}{(x+y)^{2}}
	 		= O (\ln L)
	\end{align}
	as $L\to \infty$.
\end{proof}

\begin{lemma}
 	\label{lem:P-diff-1norm}
	Assume the hypotheses of Theorem~{\upshape\ref{thm:diff}}. Then we have 
	\begin{equation}
 		\big|\tr \{ P_{L} - P_{L,0} \} \big| = O(\ln L) \qquad \text{as } L \to \infty.
	\end{equation}
\end{lemma}

\begin{proof}
	The argument is similar to the one for the previous lemma. First, we claim that the estimate \eqref{eq:MS2.25} 
	also holds if the Hilbert--Schmidt norm is replaced by the trace norm; only the constant $c_{2}$ changes: there exists $c_{1} \equiv c_{1}(E,V)> 0$ such that 
	\begin{equation}
		\label{eq:MS2.25-1}
 		\big\|1_{\Gamma_{n}}\big( 1_{<E}(H) - 1_{<E}(H_{0})\big) 1_{\Gamma_{m}}\big\|_{1}
		\le \frac{c_{1}}{(|n||m|)^{(d-1)/2} (|n| + |m|)} 
	\end{equation}
	for all $n,m \in \Z^{d} \setminus [-\ell,\ell]^{d}$, with the same $\ell >0$ as in \eqref{eq:MS2.25}. 
	We prove \eqref{eq:MS2.25-1} below. 
	Assuming its validity for the time being, we proceed with the proof of the lemma and estimate
	\begin{align}
 		\big|\tr \{ P_{L} - P_{L,0} \} \big| 
		& = \Bigg| \sum_{n\in\Z^{d}:\, \Gamma_{n} \cap \Lambda_{L} \neq\emptyset} 
				\tr\Big\{ 1_{\Gamma_{n}} 1_{\Lambda_{L}}  \big( 1_{<E}(H) - 1_{<E}(H_{0})\big) 1_{\Gamma_{n}} \Big\}\Bigg|
				\notag\\
		& \le \sum_{n\in\Z^{d}:\, \Gamma_{n} \cap \Lambda_{L} \neq\emptyset} 
				\Big\| 1_{\Gamma_{n}} 1_{\Lambda_{L}}  \big( 1_{<E}(H) - 1_{<E}(H_{0})\big) 1_{\Gamma_{n}} \Big\|_{1}
				\notag\\
		& \le C + \sum_{\substack{n\in\Z^{d} \setminus [-\ell,\ell]^{d}: \\[.5ex] \Gamma_{n} \cap \Lambda_{L} \neq \emptyset}}  \frac{c_{1}}{2|n|^{d}} = O (\ln L) \qquad\text{as $L\to \infty$.}
	\end{align}
	Here, $C \equiv C(E, V) >0$ is a constant, and we used \eqref{eq:pNorm-bound} and 
	\eqref{eq:MS2.25-1} for the second inequality.
	
	It remains to show \eqref{eq:MS2.25-1}. To this end, we recall the pointwise estimate from Eq.\ (2.7) in Ref.~\onlinecite{MuSc20}
	for the integral kernel $G_{0}(x,y;z)$ of the unperturbed resolvent $\frac{1}{H_{0} - z}$, $z\in\C\setminus\R$, 
	and apply it to
	\begin{align}
 		\Big\| |V|^{1/2} \frac{1}{H_{0} - z} 1_{\Gamma_{n}} \Big\|_{2}
		& = \bigg(\int_{\R^{d}} \d x \int_{\Gamma_{n}}\d y \, |V(x)| |G_{0}(x,y;z)|^{2} \bigg)^{1/2} \notag\\
		& \le C \bigg(\int_{\R^{d}}\d x\, |V(x)| \bigg)^{1/2} \, |z|^{(d-3)/4} \, 
					\frac{\e^{-|\Im\sqrt{z}||n|/2}}{|n|^{(d-1)/2}},
		\label{eq:MS-Lemma2.1-2}
	\end{align}
	where $n\in\Z^{d} \setminus [-\ell_{0}, \ell_{0}]^{d}$ for some suitable $\ell_{0} \equiv \ell_{0}(z,V)>0$, which is 
	given in Lemma~2.1 of Ref.~\onlinecite{MuSc20}, and $C>0$, which depends only on the dimension. 
	The bound \eqref{eq:MS-Lemma2.1-2} is analogous to the statement of 
	Lemma 2.1 in Ref.~\onlinecite{MuSc20}, except that \eqref{eq:MS-Lemma2.1-2} involves the $2$-norm instead of the $4$-norm. 
	A double application of the resolvent identity yields 
	\begin{multline}
		\label{eq:MS2.15-1}
 		\bigg\| 1_{\Gamma_{n}} \bigg( \frac{1}{H_{0} - z} - \frac{1}{H - z}\bigg)1_{\Gamma_{m} }\bigg\|_{1}
		\le \Big\|  1_{\Gamma_{n}} \frac{1}{H_{0} - z} |V|^{1/2} \Big\|_{2} \; 
		 	\Big\| |V|^{1/2} \frac{1}{H_{0} - z}  1_{\Gamma_{m}} \Big\|_{2} \\
		\times\bigg(1+ \Big\| |V|^{1/2} \frac{1}{H - z} |V|^{1/2} \Big\|\bigg).
	\end{multline}
	We observe that \eqref{eq:MS2.15-1} and \eqref{eq:MS-Lemma2.1-2} are fully analogous to 
	Eq.\ (2.15) and Lemma 2.1 in Ref.~\onlinecite{MuSc20}. 
	Thus, \eqref{eq:MS2.25-1} follows from \eqref{eq:MS2.15-1}
	in the very same way as Eq.\ (2.25)  in Ref.~\onlinecite{MuSc20} follows from 
	Eq.\ (2.15) in Ref.~\onlinecite{MuSc20} .
\end{proof}

\begin{remarks}
\item
	We expect that the trace in the claim of Lemma~\ref{lem:P-diff-1norm} remains bounded as $L\to\infty$. 
 	In fact, in their study of the spectral shift function, Kohmoto, Koma and Nakamura prove 
	the existence of the limit 
	\begin{equation}
 		\lim_{L\to\infty} \tr\big\{\vartheta_{L} \big( 1_{<E}(H) - 1_{<E}(H_{0}) \big) \vartheta_{L}\big\}
	\end{equation}
	in Thm.~5 of Ref.~\onlinecite{KoKoNa13}, where $\vartheta_{L} := \vartheta(L^{-1} \,\boldsymbol\cdot\,)$ is a smooth cut-off function determined by some 
	$\vartheta \in C_{c}^{\infty}(\R^{d})$ which obeys $\vartheta=1$ in a neighbourhood of the origin.
	Unfortunately, we cannot use this result in our study of the Szeg\H{o} asymptotics with ``sharp'' indicator functions because the volume of the region where $\vartheta_{L}$ decays from 1 to 0 is of the order $O(L^{d})$.
\item
	\label{rem:1Dfail}
 	Likewise, we believe that the trace on the r.h.s. of the first line of \eqref{eq:norm-est-bad} 
	remains bounded
	as $L\to\infty$. On the other hand, the trace norm $\|(P_{L} - P_{L,0}) s_{n}(P_{L,0})\|_{1}$ 
	may grow logarithmically in $L$ in some one-dimensional situations, which is correctly captured by the product of Hilbert--Schmidt norms in the second line of \eqref{eq:norm-est-bad}. Thus, it is the second inequality in \eqref{eq:norm-est-bad} which is responsible for the additional symmetry constraint in $\HH_{1}$, resp.\ $\HH_{1,0}$, as compared to higher dimensions.
\end{remarks}

The difference $Q_{L} - Q_{L,0}$ was previously estimated 
in Lemma 2.5 of Ref.~\onlinecite{MuSc20} in suitable von Neumann--Schatten norms 
$\|\boldsymbol\cdot\|_{s}$. The next lemma improves that result by accessing smaller values of $s$ and obtaining a weaker growth in $L$ as compared to Lemma 2.5 of Ref.~\onlinecite{MuSc20}.

\begin{lemma}
	\label{lem:interpolation}
	Let $\Lambda\subset\R^{d}$, $d\in\N$, be as in Assumption~{\upshape\ref{ass:l}(ii)}. 
	Let $V\in L^\infty(\R^d)$ have compact support and fix $E>0$ and  $s\in{}]0,1]$. 
	Then, for every $p_{0} \in {}]0,\min\{1,2s\}]$ there exists a constant $C\equiv C(\Lambda,V,E,p_{0})>0$ 
	such that for all $ L\ge 1$ we have
	\begin{equation}
		\label{eq:FermiProjDivs}
		\| Q_{L} - Q_{L,0} \|_{2s}^{2s}\le C L^{2d(1-s)/(2-p_{0})}.
	\end{equation}
\end{lemma}

\begin{proof}
 	We fix $E>0$, $s\in{}]0,1]$ and $p_{0} \in {}]0,\min\{1,2s\}]$. 
	Given any $L>1$, the spatial domain $\Lambda_{L} \subseteq \bigcup_{n \in\Xi_{L}} \Gamma_{n}$ can be covered 
	by finitely many unit cubes indexed by the set $\Xi_{L} \equiv \Xi_{L}(\Lambda) \subset \Z^{d}$. 
	The number of required cubes can be bounded  by $|\Xi_{L}| \le \varkappa L^{d}$ where the constant 
	$\varkappa \equiv \varkappa(\Lambda) >0$ does not depend on $L$. We therefore conclude from \eqref{eq:pNorm-bound}
	that
	\begin{equation}
		\label{eq:norm-pNull}
 		\| Q_{L} - Q_{L,0} \|_{p_{0}}^{p_{0}}
		\le 2 \gamma \varkappa L^{d}.
	\end{equation}
	We will apply the interpolation inequality 
	\begin{equation}
		\label{eq:interpol}
 		\|\boldsymbol\cdot\|_{p_{\theta}} \le \|\boldsymbol\cdot\|_{p_{0}}^{1-\theta}
			 \|\boldsymbol\cdot\|_{p_1}^{\theta} \qquad\text{with}\quad
		\frac{1}{p_{\theta}} = \frac{\theta}{p_{1}} + \frac{1-\theta}{p_{0}}
	\end{equation}
	to von Neumann--Schatten (quasi-) norms. It is valid for every $0 < p_{0} \le p_{1} < \infty$ and every $\theta \in [0,1]$.
	We choose $p_{\theta} := 2s$, $p_{1} := 2$ and determine 
	\begin{equation}
 		\theta = 1 - \frac{p_{0}}{s}\, \frac{1-s}{2-p_{0}}.
	\end{equation}
	Together with the boundedness of Hilbert--Schmidt norms \eqref{eq:HSbounded},
	we infer the claim from \eqref{eq:interpol} and \eqref{eq:norm-pNull}.	
\end{proof}

\begin{corollary}
	\label{cor:interpolation}
	Let $d\in \N\setminus\{1\}$ and $\Lambda\subset\R^{d}$ be as in Assumption~{\upshape\ref{ass:l}(ii)}.
	Let $V\in L^\infty(\R^d)$ have compact support and fix $E>0$ and  $s\in{}]d^{-1},1]$. 
	Then we have
	\begin{equation}
		\label{eq:FPDiffsmall-o}
		\| Q_{L} - Q_{L,0} \|_{2s}^{2s} = o(L^{d-1})
		\quad\text{as } L\to\infty.
	\end{equation}
\end{corollary}

\begin{proof}
	For $s\in {}]0,1]$ and $p_{0}\in {}]0,2[\,$ we observe
	\begin{equation}
 		2d (1-s) /(2-p_{0}) < d-1 \quad \Longleftrightarrow \quad
		s > d^{-1} + p_{0} (d-1)/(2d).
	\end{equation}
	Now, the claim follows from Lemma~\ref{lem:interpolation} because $p_{0}$ can be chosen arbitrarily close to zero.
\end{proof}

\begin{remark}
 	We summarise the results of this section in the following structural reformulation of 
	Theorem~\ref{thm:diff} which is independent of the concrete forms of the unperturbed operator $H_{0}$ and of the perturbed operator $H$: \\[1ex]
	\emph{Let $H_{0}$ and $H$ be two self-adjoint operators that are bounded below. Consider a Fermi energy $E\in\R$, 
	a bounded subset $\Lambda \subset \R^{d}$, 
	coefficients $N_{0}(E,\Lambda)>0$ and $\Sigma_{0}(E,\Lambda)>0$ 
	and a linear functional $h \mapsto I(h)$, defined on test functions $h$ as specified above \eqref{eq:widom}
	and depending on $h$ only through 
	$h - \mathrm{id}\,h(1)$ and such that 
	\begin{align}
 	\bullet\quad & I(h) \ge 0 \quad \text{whenever $h \ge 0$ and $h(1)=0$,} \notag \\
	\bullet\quad & \lim_{n\to\infty} I(h_{n}) = 0 \quad \text{whenever~~} 
		\lim_{n\to\infty} \int_{0}^{1} \d\lambda \, |h_{n}(\lambda)| = 0 
		\;\,\text{and $h_{n}\big|_{\mathcal{N}}=0$ } \notag \\ 
	& \hspace{2cm} \text{for all $n\in\N$ with an $n$-independent neighbourhood $\mathcal{N}$ of both $0$ and $1$.} \notag
	\end{align}
	Suppose that, given any test function $h$ as specified above \eqref{eq:widom}, the asymptotics  
	\begin{equation}
		\label{eq:widom-abs}
 		\tr\big\{h\big(1_{\Lambda_L}1_{<E}(H_0)1_{\Lambda_L}\big)\big\}
		= N_{0}(E,\Lambda) h(1) L^{d} + \Sigma_{0}(E,\Lambda) I(h) \,  L^{d-1} \ln L 
			+ {o}\big(L^{d-1} \ln L\big) 
	\end{equation}
	holds as $L\to\infty$.	
	Assume further that the a-priori-norm bound \eqref{eq:pNorm-bound} holds for both $H_{0}$ and $H$. Finally, 
	suppose the validity of \eqref{eq:HSbounded} and of the conclusions of Lemmas~\ref{lem:P-diff-2norm} and~\ref{lem:P-diff-1norm}. 
	Then,  for every test function $h \in \HH_{d}$, we have 
	\begin{equation}
	 	\tr\Big\{h\big(1_{\Lambda_L}1_{<E}(H)1_{\Lambda_L}\big) - 
			h\big(1_{\Lambda_L}1_{<E}(H_{0}) 1_{\Lambda_L}\big)\Big\} 
		= {o}(L^{d-1} \ln L) 
	\end{equation}
	as $L\to\infty$.
	}
\end{remark}

\newcommand{\noopsort}[1]{} \newcommand{\singleletter}[1]{#1}
%


\end{document}